\documentclass[reqno]{amsart}
\usepackage[utf8]{inputenc}
\usepackage{amsmath,amsthm,amsfonts,amscd,amssymb,amstext}
\usepackage{booktabs}
\usepackage{esint}
\usepackage{graphicx,color}
\usepackage{tabularx}
\usepackage{titletoc}
\usepackage{microtype}
\usepackage[colorlinks=true,linkcolor=black,citecolor=black,urlcolor=blue]{hyperref}
\usepackage{listings}
\usepackage{cite}
\usepackage{url}
\usepackage{alltt}

\theoremstyle{plain}
\newtheorem{theorem}{Theorem}

\newtheorem{prop}{Proposition}
\newtheorem{cor}{Corollary}

\newtheorem{example}{Example}

\theoremstyle{remark}
\newtheorem*{rmk}{Remark}
\newtheorem*{claim}{Claim}

\newcommand{\barr}{\begin{eqnarray}}
\newcommand{\earr}{\end{eqnarray}}
\newcommand{\be}{\begin{equation}}
\newcommand{\ee}{\end{equation}}

\def\Ord{\mathcal{O}}
\newcommand{\de}{\mathrm{d}}

\renewcommand{\d}[2]{\frac{d #1}{d #2}} 

\newcommand{\ket}[1]{\left| #1 \right>} 
\let\baraccent=\= 
\renewcommand{\=}[1]{\stackrel{#1}{=}} 

\newcommand{\numberset}{\mathbb}
\newcommand{\N}{\numberset{N}}

\newcommand{\R}{\numberset{R}}
\newcommand{\C}{\numberset{C}}

\newcommand{\Tr}{\mathrm{Tr}}

\newcommand{\diag}{\mathrm{diag}}

\newcommand{\as}{\mathrm{a.s.}}
\newcommand{\iid}{\mathrm{i.i.d.}}

\newcommand{\E}{\mathbf{E}}

\newcommand{\HH}{\mathcal{H}}

\newcommand{\hc}{\mathrm{h.c.}}

\def\d{\dagger}

\begin{document}

\title[Quadratic forms of Fermi operators]{Density and spacings for the energy levels of quadratic Fermi operators}

\author[Cunden]{Fabio Deelan Cunden}
\address[Fabio Deelan Cunden]{School of Mathematics, University of Bristol, University Walk, Bristol BS8 1TW, United Kingdom}
\author[Maltsev]{Anna Maltsev}
\address[Anna Maltsev]{School of Mathematical Sciences, 
Queen Mary University of London, 
London E1 4NS, 
United Kingdom}
\author[Mezzadri]{Francesco Mezzadri}
\address[Francesco Mezzadri]{School of Mathematics, University of Bristol, University Walk, Bristol BS8 1TW, United Kingdom}

\date{\today}
\maketitle

\begin{abstract}
The work presents a proof of convergence of the density of energy levels to a Gaussian distribution for a wide class of quadratic forms of Fermi operators. This general result applies also to quadratic operators with disorder, e.g., containing random coefficients. The spacing distribution of the unfolded spectrum is investigated numerically. For generic systems the level spacings behave as the spacings  in a Poisson process. Level clustering persists in presence of disorder.
\end{abstract}

\tableofcontents

\section{Introduction}
In a variety of situations one encounters quadratic forms in Fermi operators
\be
\HH_n=\sum_{i,j=1}^n A_{ij}{c_i}^\d c_j+\frac{1}{2}B_{ij}(c_i c_j-{c_i}^\d {c_j}^\d),\label{eq:quadratic}
\ee
where the Fermi operators $c_i$'s obey the canonical anticommutation relations $\{c_i,c_j\}=0$, $\{c_i,{c_j}^\d\}=\delta_{ij}$, and the coefficients satisfy $A_{ij}=A_{ji}\in\R$, $B_{ij}=-B_{ij}\in\R$ for $ i,j=1,2,\dots$.
The quadratic form~\eqref{eq:quadratic} defines a symmetric operator $\HH_n$ acting on a Hilbert space of dimension $2^n$. This operator represents the Hamiltonian of a system of quasifree fermions.

Such quadratic operators are of fundamental interest for several reasons. First and foremost, these operators can be diagonalized exactly using an explicit normal modes decomposition~\cite{Lieb61} and certain quadratic Hamiltonians are good approximations for more complicated two-body interactions. Furthermore, fermionic models share a close relationship  to interacting spins in dimension one,  and they are among the simplest systems in which quantum phase transitions occur and entanglement measures can be computed. The literature on quasifree fermions, their relationship to spin systems, as well as to other areas of physics such as conformal field theory and random matrix theory is immensely vast. For a review, see~\cite{Amico08,Calabrese09,Eisert10,Keating04}. In the last decades, models with disordered, i.e., containing random parameters, have also been considered~\cite{Fisher94,Refael04,Prosen08,Pastur14}.

For quadratic forms in Fermi operators it is of interest to know whether, in the limit of large $n$, the density of energy levels and the spacing distribution  converge and to identify the limit. Of course, the spacing distribution of the unfolded spectrum requires knowledge of the density of energy levels. We discuss these questions for a broad class of systems that includes the following examples.

\begin{enumerate}
\item[(i)] The standard diagonalisation scheme for spin $1/2$ systems is centred on the Jordan-Wigner transformation, which is based on the observation that there exists a unitary mapping between the Hilbert space $(\C^2)^{\otimes n}$ of $n$  spin $1/2$'s and the antisymmetric Fock space $\mathcal{F}_{-}(\C^n)$ of spinless fermions on $n$ sites. For instance, the $XY$ model consists of $n$ spin $1/2$'s ($n$ even) arranged in a chain and having only nearest neighbour interactions $(1/2)\sum_i(1+\gamma)\sigma_i^x\sigma_{i+1}^x+(1-\gamma)\sigma_i^y\sigma_{i+1}^y$, where $\sigma^x$ and $\sigma^y$ may be represented by the usual Pauli matrices ($\hbar=1$). Using the Jordan-Wigner map, the model can be cast as a quadratic form in Fermi operators~\cite{Lieb61}
\be
\HH_n^{XY}=\sum_{i=1}^{n-1}[({c_i}^\d c_{i+1}+\gamma c_i^\d {c_{i+1}}^\d)+\hc].\label{eq:HXYfermions}
\ee
This representation is exact in the case of spin chain with free ends.
\item[(ii)] A quantum bond percolation model on a lattice $\Gamma$ with $n$ sites consists of a tight-bind Hamiltonian of the form \cite{Berkovits96,Sen09,Schmidtke14}
\be
\HH_n^{perc}=\sum_{<ij>}[t_{ij}{c_i}^\d c_{j}+\hc],\label{eq:percolation}
\ee
where the summation runs over nearest neighbour sites and the hopping matrix elements $t_{ij}\in\R$ are independent Bernoulli random variables $\Pr(t_{ij}=1)=1-\Pr(t_{ij}=0)=p\in(0,1)$.
\item[(iii)] The Anderson model \cite{Anderson58} is one of the simplest models incorporating the essential competition between the
hopping term (discrete Laplacian) and the on-site disorder (random potential). For a generic lattice $\Gamma$ the Anderson Hamiltonian for noninteracting fermions can be written
\be
\HH_n^{And}=\sum_{i}v_ic_i^\d c_i+t\sum_{<ij>}[{c_i}^\d c_{j}+\hc],\label{eq:Anderson}
\ee
with random on-site potential $v_i\in\R$; usually $v_i$'s are independent with mean zero and finite variance $W^2$.
\item[(iv)] More general non-sparse random quasifree fermions Hamiltonian. For instance one can consider the Hamiltonian~\eqref{eq:quadratic} with $A_{ij}$, $B_{ij}$ independent Gaussian variables, modulo the symmetries $A_{ij}=A_{ji}$ and $B_{ij}=-B_{ji}$. It turns out that this model is related to the real Ginibre ensemble of random matrix theory~\cite{Ginibre65}. 
\end{enumerate}

In the traditional paradigm of condensed matter physics, the number of particles is so large that questions on the macroscopic density of energy levels, i.e., the behaviour of the energy levels in the `bulk' very far from the ground state, are meaningless. The situation has changed recently. Over the past few years, experimental developments have allowed the study of systems  with a small and controlled number of particles
and therefore, a direct measure of the level density might be within reach of
current experimental capabilities.

These considerations have recently triggered the attention of some authors on the problem of convergence and universality of the limiting level density of many body systems. 
In two pioneering papers, Hartmann, Mahler and Hess~\cite{Hartmann04} considered generic many body quantum systems with nearest neighbour interaction. They proved that, provided that the energy per particle has an upper bound, the energy distribution for almost every product state becomes a Gaussian in the limit of infinite number of particles.
More recently, Atas and Bogomolny~\cite{Atas14,Atas14b} investigated numerically and theoretically the energy levels of several interacting spin $1/2$ systems and concluded that the density of levels converges to a Gaussian. Using an adaptation
of the line of reasoning in~\cite{Hartmann04}, Keating, Linden and Wells~\cite{Keating14,Keating15,Wells14} proved convergence to a Gaussian distribution for spin chains with generic pair interactions, including the case of spin glasses, i.e., interaction with random couplings. This result has been extended to spin systems on more general graphs by Erd\"{o}s and Schr\"{o}der~\cite{Erdos14}. 
The algebraic identities satisfied by the Pauli matrices representing spin $1/2$'s play a key role in the proofs in~\cite{Keating14,Wells14,Keating15,Erdos14}.

Our goal here is to show that the density of energy levels of a wide class of quadratic Fermi operators (both deterministic and random) converges to a Gaussian distribution in the limit of large $n$. The proof of this universal result relies on the connection between the spectrum of $\HH_n$ and the subset-sum structure arising in the normal modes decomposition. This result explains some of the previous conjectural statements and numerical observations by Atas and Bogomolny on the level density of certain (nonrandom) spin systems. Additionally we provide a uniform bound (based on a Berry-Esseen inequality) on the rate of convergence. 

 We also consider the level spacing distribution of such operators. Numerical investigation shows that both deterministic and random models exhibit level clustering (Poisson statistics); this behavior is compatible with the celebrated Berry-Tabor philosophy for generic integrable systems~\cite{Berry77}, even in presence of disorder.
In the course of the paper we also present a few short examples illustrating the general theorems.

The paper is organised as follows. In Section~\ref{sec:setup} we set the notation and review the consequences of the normal modes decomposition. Then, in Section~\ref{sec:results} we present our main results on the limiting density of energy levels and the rate of convergence to the limit. In Sections~\ref{sec:spin}, \ref{sec:perc}, \ref{sec:Ginibre} and \ref{sec:band} we apply the  general theorems to the examples (i), (ii), (iii) and (iv) discussed above, thus illustrating in physical models the universality of Theorems~\ref{thm:DEL} and Corollary~\ref{thm:DELrandom}. Finally, in Section~\ref{sec:spacings} we present the numerical observations on the level spacing distribution.

\emph{Notation.}  We shall denote by $r_k$ ($k=1,2,\dots$) a collection of $\iid$ binary variables with $\Pr(r_k=1/2)=\Pr(r_k=-1/2)= 1/2$. Expectation with respect to the $r_k$'s  will be denoted by $\mathbf{E}[\cdot]$. By  $\displaystyle\|\cdot\|_{\mathrm{op}}$ we shall indicate the usual operator norm (the largest singular value). The projection onto the first $n$ coordinates will be denoted by $P_n=\diag(\underbrace{1,1,\dots,1}_{\text{$n$ times}},0,0,\dots)$.

\section{Generalities on Fermi operators}\label{sec:setup}
Let us order the $2^n$ eigenvalues of $\HH_n$ as
\be
E_{1,n}\leq E_{2,n}\leq\cdots\leq E_{2^n-1,n}\leq E_{2^n,n}.\label{eq:ordH}
\ee
For a quadratic Hamiltonian~\eqref{eq:quadratic} it is possible to write a normal modes decomposition. More precisely, using a canonical transformation~\cite[Appendix A]{Lieb61} the operator $\HH_n$ can be written as
\be
\HH_n=\sum_{k=1}^n\lambda_{k,n}\left(\eta_k^\d\eta_k-\frac{1}{2}\right)+K_n, \label{eq:normalmodes} 
\ee
where the \emph{normal modes} $\eta_k,\eta_k^\d$ are Fermi operators, the \emph{elementary excitations} $\lambda_{k,n}\geq0$ are the singular values of $P_n(A+B)P_n$ and $K_n=\Tr P_nAP_n/2$.

The following well-known properties of the Fermi operators $\eta_k$, $\eta_k^\d$ are immediate consequences of the canonical anticommutation relations~\cite{Nielsen05}. First, the $\eta_k^\d\eta_k$ are Hermitian operators with eigenvalues $0$ and $1$. Second, $\eta_k$ ($\eta_k^\d$) acts as a lowering (raising) operator on the normalised eigenvectors of $\eta_k^\d\eta_k$ with eigenvalue $1$ ($0$). Moreover, the $\eta_k^\d\eta_k$'s form a set of mutually commuting operators and therefore they can be simultaneously diagonalised. These three facts imply that there exists a normalised vector $\ket{0}$ (the vacuum state) which is an eigenvector of all the $\eta_k^\d\eta_k$'s with corresponding eigenvalue zero:  $\eta_k^\d\eta_k\ket{0}=0$. A set of $2^n$ normalised eigenvectors of $\eta_k^\d\eta_k$ ($k=1,\dots,n$) can be built up by exciting the vacuum state; the normalised vector $\ket{\alpha_1\alpha_2\cdots\alpha_n}=(\eta_1^\d)^{\alpha_1}(\eta_2^\d)^{\alpha_2}\cdots(\eta_n^\d)^{\alpha_n}\ket{0}$ with $\alpha_k=0$ or $1$, is an eigenvector of $\eta_k^\d\eta_k$ with eigenvalue $\alpha_k$. Therefore, from~\eqref{eq:normalmodes}  we have
\be
\HH_n\ket{\alpha_1\alpha_2\cdots\alpha_n}=\left(K_n+\sum_{k=1}^n\alpha_k\lambda_{k,n}-\frac{1}{2}\sum_{k=1}^n\lambda_{k,n}\right)\ket{\alpha_1\alpha_2\cdots\alpha_n}.
\ee
The spectrum of $\HH_n$ is constructed by exciting the ground state energy $E_{1,n}=K_n-1/2\sum_{k}{\lambda_{k,n}}$ by the elementary excitations $\lambda_{k,n}$.
Hence the spectrum is characterised in terms of the \emph{subset sums} of elementary excitations as follows: $E$ is an eigenvalue of $\HH_n$ if and only if
\be
\exists S\subseteq\{1,\dots,n\}\quad \text{such that}\quad E=K_n+\frac{1}{2}\left(\sum_{k\in S}\lambda_{k,n}-\sum_{k\notin S}\lambda_{k,n}\right).\label{eq:sp(H)}
\ee
The density of energy levels is defined as the empirical normalised measure
\be
\frac{1}{2^n}\sum_{k=1}^{2^n}\delta(E-E_{k,n}),\label{eq:counting1}
\ee
and from~\eqref{eq:sp(H)} it follows that
\be
\frac{1}{2^n}\sum_{k=1}^{2^n}\delta(E-E_{k,n})=\frac{1}{2^n}\sum_{r_1,\dots,r_n\in\{\pm \frac{1}{2}\}}\delta\left(E-\sum_{k=1}^n r_k\lambda_{k,n}-K_n\right).\label{eq:emp2}
\ee
Up to a shift, the empirical measure of $E_{k,n}$ is given by the distribution of the sum of $n$ independent variables $r_1\lambda_{1,n},\dots,r_n\lambda_{n,n}$.
In fact, it is possible to compute the Fourier transform of~\eqref{eq:emp2}:
\barr
\int\frac{1}{2^n}\sum_{k=1}^{2^n}\delta(E-E_{k,n})e^{itE}\de E&=&e^{i tK_n}\prod_{k=1}^n\E[e^{i t r_k\lambda_{k,n}}]\nonumber\\
&=&e^{i tK_n}\prod_{k=1}^n\left(\frac{1}{2}e^{it \lambda_{k,n}/2}+\frac{1}{2}e^{-it\lambda_{k,n}/2}\right)\nonumber\\
&=&e^{i tK_n}\prod_{k=1}^n\cos\left(\frac{t\lambda_{k,n}}{2}\right).\label{eq:comp_fourier}
\earr
This computation shows that the empirical distribution of the energy levels $E_k$ is the distribution of a sum of independent random variables.
It is then plausible that for large $n$, after a suitable rescaling, the distribution of  energy levels converges to a Gaussian. After all, the many body Hamiltonian~\eqref{eq:normalmodes} is a sum of single particle (commuting) operators and the total spectrum is given by the sum of the individual spectra. In the following section we specify exact conditions for this convergence. Note that the variables $r_k\lambda_k$'s are independent but not identically distributed, e.g. $\E[r_k\lambda_{k,n}]=0$ and $\E[(r_k\lambda_k)^2]=\lambda_{k,n}^2/4$. 

Before stating the main theorems we conclude this section with a last computation to prepare the ground to what follows. If we knew that the limiting level density is Gaussian then the limit would be identified by its mean and variance. The moments of the counting measure~\eqref{eq:counting1} are related to traces of powers of $\HH_n$ by the following identity
\be
\frac{1}{2^n}\Tr \HH_n^p=\int\frac{1}{2^n}\sum_{k=1}^{2^n}\delta(E-E_{k,n})E^p \de E.
\ee
In particular, mean and variance are given by the traces of the first two powers $\Tr\HH_n$ and $\Tr\HH_n^2$. A direct computation of these traces is possible using Wick's calculus. The only non-traceless products of Fermi operators that we need are
\barr
\Tr(c_i^\d c_j)&=&2^{n-1}\delta_{ij},\label{eq:id1}\\
\Tr(c_i^\d c_jc_k^\d c_l)&=&2^{n-2}(\delta_{ij}\delta_{kl}+\delta_{il}\delta_{jk}),\label{eq:id2}\\
\Tr(c_i c_jc_k^\d c_l^\d)&=&2^{n-2}(\delta_{il}\delta_{jk}-\delta_{ik}\delta_{jl})\label{eq:id3};
\earr
using~\eqref{eq:id1}-\eqref{eq:id3} one finds
\be
\frac{1}{2^n}\Tr\HH_n=\frac{1}{2}\sum_{i=1}^nA_{ii},\,\,\text{and}\,\,\,
\frac{1}{2^n}\left(\Tr\HH_n^2-(\Tr\HH_n)^2\right)=\frac{1}{4}\sum_{i,j=1}^n(A_{ij}^2+B_{ij}^2).\label{eq:finitevar}
\ee
The above quantities  are mean and variance of the finite-$n$ level density.
\section{Main results}\label{sec:results}
\begin{theorem}[Density of energy levels]\label{thm:DEL} Let $\HH_n$ be the quadratic form~\eqref{eq:quadratic}. Assume that, denoting $X_n=\displaystyle P_n(A+B)P_n$, the following conditions are true:
\begin{itemize}
\item[i)] $\displaystyle\lim_{n\to\infty}n^{-1/4}\displaystyle\|X_n\|_{\mathrm{op}}=0$;
\item[ii)] $\displaystyle\lim_{n\to\infty}\frac{1}{4n}\Tr (X_n^{T}X_n)=\sigma^2<\infty$.
\end{itemize}
Then, the density of shifted and rescaled energy levels
\be
\nu_n(E)=\frac{1}{2^n}\sum_{k=1}^{2^n}\delta\left(\frac{E_{k,n}-K_n}{\sqrt{n}}-E\right)
\ee
weakly converges, as $n\to\infty$,  to a centred Gaussian probability measure with  variance $\sigma^2$:
\be
\de\nu_n(E)\rightharpoonup\frac{1}{\sqrt{2\pi\sigma^2}}e^{-\frac{E^2}{2\sigma^2}}\de E. \label{eq:them_conv}
\ee
\end{theorem}
Theorem~\ref{thm:DEL} can be proved by checking the Feller-Lindeberg conditions~\cite{Feller71} in the central limit theorem for independent nonidentical random variables. We present however a more direct proof based on elementary computations.
Hypothesis i) of Theorem~\ref{thm:DEL} can be rephrased as
\be
\lambda_{k,n}=o(n^{1/4})\quad\text{for all $k$},
\ee
meaning that the elementary excitations do not grow too fast with $n$. This assumption is similar (but in sense weaker) to the condition of finite energy per particle in Hartmann, Mahler and Hess theorem~\cite{Hartmann04}.  Note also that
\be
\sigma^2=\lim_{n\to\infty}\frac{1}{4n}\Tr (X_n^{T}X_n)=\lim_{n\to\infty}\frac{1}{4n}\sum_{k=1}^n\lambda^{2}_{k,n}=\lim_{n\to\infty}\frac{1}{4n}\sum_{i,j=1}^n(A_{ij}^2+B_{ij}^2),
\ee
according to~\eqref{eq:finitevar}. Hypothesis ii) is thus a condition on the second moment of the density of energy levels.
\begin{proof}[Proof of Theorem~\ref{thm:DEL}]  Let $\lambda_{k,n}$ ($k=1,\dots,n$) be the singular values of $X_n$. Repeating the computation in~\eqref{eq:comp_fourier} we find
\be
\int e^{itE}\de \nu_n(E)=\prod_{k=1}^n\cos\left(\frac{t\lambda_{k,n}}{2\sqrt{n}}\right). \label{eq:proof1}
\ee
The key point to appraise~\eqref{eq:proof1} is the following identity.
\begin{claim}
Let $(u_i)_{i\in\N}$ be a sequence of complex numbers such that
\be
\lim_{n\to\infty}n^{-1/2}\max_{1\leq i\leq n}|u_i|=0, \label{eq:assmp1}
\ee
 and the following limit exists and is finite
\be
S=\lim_{n\to\infty}\frac{1}{n}\sum_{i=1}^{n}u_i.\label{eq:assmp2}
\ee
Then
\be
\lim_{n\to\infty}\prod_{i=1}^{n}\left(1+\frac{u_i}{n}\right)=e^S.\label{eq:prod_exp}
\ee
(A generalization of the identity $\lim\limits_{n\to\infty}\left(1+\frac{u}{n}\right)^n= e^u$.)
\end{claim}
If we accept the claim, we can prove the theorem as follows. For any fixed $t\in\R$:
\barr
\prod_{k=1}^n\cos\left(\frac{t\lambda_{k,n}}{2\sqrt{n}}\right)
&=&\prod_{k=1}^n\left(1-\frac{1}{2}\frac{t^2\lambda_{k,n}^2}{4n}(1+o(1))\right).
\earr
By the claim and hypotheses i) and ii) (using $\max_k\lambda_{k,n}^2=\displaystyle\|X_n^TX_n\|_{\mathrm{op}}$) the last expression converges to $\exp(-\sigma^2t^2/2)$ 
and by L\'evy's continuity theorem this proves~\eqref{eq:them_conv}.

It remains to prove the claim. We adapt a proof given in~\cite[Lemma A.5]{Maltsev16}. Set
\be
P_n=\prod_{i=1}^{n}\left(1+\frac{u_i}{n}\right), \quad S_n=\frac{1}{n}\sum_{i=1}^{n}u_i,\quad M_n=\max_{1\leq i\leq n}|u_i|.
\ee
Note that the function $\log(1+z)+z$ has a double zero at $z=0$. Hence 
\be
L(z)= (\log(1+z)+z)/z^2
\ee 
is analytic in the open disk $|z|<1$ (in particular it is continuous and bounded). The finite product can be written as
\be
P_n=e^{S_n}\exp\left\{\sum_{i=1}^n\left(\frac{u_i}{n}\right)^2L\left(\frac{u_i}{n}\right)\right\}.
\ee
Therefore
\be
P_n-e^{S_n}=e^{S_n}\left(\exp\left\{\sum_{i=1}^n\left(\frac{u_i}{n}\right)^2L\left(\frac{u_i}{n}\right)\right\}-1\right).
\ee
By continuity there exists $0<R<1$ such that $L(R)=1$. From~\eqref{eq:assmp1} it follows that, for $n$ sufficiently large, $M_n/n\leq R$ and by the maximum principle $|L(\frac{u_i}{n})|\leq1$. We conclude that, for large $n$,
\be
\left|P_n-e^{S_n}\right|\leq |e^{S_n}|\frac{M_n^2}{n}e^{\frac{M_n^2}{n}}\label{eq:ineq}
\ee
since  for any $z$, $|e^z-1|\leq|z|e^{|z|}$. By~\eqref{eq:assmp1} and \eqref{eq:assmp2} the above inequality implies the claim~\eqref{eq:prod_exp}.
\end{proof}
The following result provides a uniform bound on the discrepancy between the  finite-$N$ empirical density of energy levels and the  limiting Gaussian (in the sense of Kolmogorov distance between probability distributions).
\begin{prop}\label{prop:BE}
Denote
\be
s_n^2=\frac{1}{4n}\Tr (X_n^{T}X_n), \quad\rho_n^3=\frac{1}{8n^{3/2}}\Tr((X_n^{T}X_n)^{3/2}).
\ee
Then for all $n$ the following bound on the distance between the counting measure of the normalised energy levels $E_{k,n}/s_n$ and the  standard Gaussian distribution holds
\be
\sup_E
\left|\frac{1}{2^n}\#\left\{k\colon \frac{E_{k,n}-K_n}{s_n}<E\right\}-\frac{1}{\sqrt{2\pi}}\int_{-\infty}^E e^{-\frac{x^2}{2}}\de x\right|
\leq \frac{C}{\sqrt{n}}\frac{\rho_n^3}{s_n^{3/2}}, \label{eq:thm_bound}
\ee
for an absolute constant $C$ that may be chosen as $C=6$.
\end{prop}
\begin{proof}
To prove~\eqref{eq:thm_bound} we use a classical Berry-Esseen inequality for independent nonidentically distributed variables.
The empirical distribution of the shifted energy levels is the same as the empirical distribution of the sum of independent centred random variables $x_1,\dots,x_n$ with the position
\be
x_k=r_k\lambda_{k,n},
\ee
where $r_1,r_2,\dots$ are i.i.d. binary variables (note that the $x_k$'s are not identically distributed). Denote by $F_n$ the cumulative distribution of the normalised sum $(x_1+\cdots +x_n)/(\sum_{k=1}^n\E[x_k^2])^{1/2}$. Then for all $x$ and $n$
\be
\left|F_n(x)-\frac{1}{\sqrt{2\pi}}\int_{-\infty}^x e^{-y^2/2}\de y\right|
\leq \frac{C}{\sqrt{n}}\frac{\sum_{k=1}^n\E|x_k|^3}{(\sum_{k=1}^n\E x_k^2)^{3/2}},
\ee
where $C\leq 6$ (see Ch.~XVI.5, Theorem 2 in~\cite{Feller71}). An elementary computation shows that
\be
\E x_k^2=\frac{\lambda_{k,n}^2}{4}\quad\text{and}\quad\E|x_k|^3=\frac{\lambda_{k,n}^3}{8},
\ee
where the expectation value is taken with respect to $r_1,\dots,r_n$. This concludes the proof, since the $\lambda_{k,n}$'s are (up to a rescaling) the singular values of $X_n$.
\end{proof}
\begin{example}\label{ex:1} We show that the rate $n^{-1/2}$ in~\eqref{eq:thm_bound} is optimal. Suppose that $A_{ij}=\xi \delta_{ij}$ ($\xi\in\R$) and $B_{ij}=0$. Hence, the quadratic form reads
\be
\HH_n=\xi\sum_{k=1}^n\eta_k^{\d}\eta_k.\label{eq:ham_ex1}
\ee
In this case the elementary excitations $\lambda_{k,n}$ are all equal to $\xi$ and empirical distribution of the energy level $E_{k,n}$ is given by the distribution of the sum of i.i.d. variables $x_k=\xi r_k$ with $r_k$ as above. Therefore we have
\be
\frac{1}{2^n}\sum_{k=1}^{2^n}\delta\left(E-\frac{E_{k,n}-\xi n/2}{\sqrt{n}}\right)=\E\delta\left(E-\frac{1}{\sqrt{n}}\sum_{k=1}^nx_k\right).
\ee
By the central limit theorem for i.i.d. random variables, $\frac{1}{\sqrt{n}}\sum_{k=1}^nx_k$ converges to a Gaussian variable with mean $0$ and variance $\xi^2/4$ (compare with Theorem \ref{thm:DEL}).  Moreover, by Chebychev inequality $\sum_{k=1}^nx_k\in(-\xi\sqrt{n},\xi\sqrt{n})$ with probability at least $3/4$ and therefore each value in this interval is taken with probability proportional to $\frac{1}{\xi\sqrt{n}}$. Hence, the distribution of the discrete random variable $\sum_{k=1}^nx_k$ has jumps of size $n^{-1/2}$. On the other hand the Gaussian distribution is continuous. So the error in the Gaussian approximation is at least given by the size of the jumps which matches with the bound in \eqref{eq:thm_bound}.
\end{example}
Theorem \ref{thm:DEL} can be adapted to deal with random quadratic Fermi Hamiltonians (see the examples (ii), (iii) and (iv) presented in the introduction).
Let $(\Omega,\mathcal{F},\mathbb{P})$ be a probability space. The expectation with respect to $\mathbb{P}$ will be denoted by $\mathbb{E}$. Let us suppose that $A(\omega)$ and $B(\omega)$  ($\omega\in\Omega$) are random double arrays of real numbers satisfying $A(\omega)_{ij}=A(\omega)_{ji}$ and $B(\omega)_{ij}=-B(\omega)_{ji}$. Hence \eqref{eq:quadratic} defines a sequence of random quadratic forms  $\HH_n(\omega)$ in Fermi operators. Our approach to proving convergence to a Gaussian limit consists of two steps: firstly, we average over fictitious binary variables (using Theorem \ref{thm:DEL}) for a given realization of the disorder ($A(\omega)_{ij}$ and $B(\omega)_{ij}$); then, if the first average in the limit of large $n$ is independent of the realization $\omega$, we can average over the disorder (i.e. with respect to $\mathbb{P}$). Note that all random variables are defined on the same probability space.
We have the following result as a corollary of  Theorem \ref{thm:DEL}.
\begin{cor}\label{thm:DELrandom} Let $\HH_n(\omega)$ be the random quadratic form~\eqref{eq:quadratic} defined by $A(\omega)$ and $B(\omega)$. Let $X_n(\omega)=\displaystyle P_n(A(\omega)+B(\omega))P_n$ and assume that the following conditions hold true $\mathbb{P}$-almost surely:
\begin{itemize}
\item[i)] $\displaystyle\lim_{n\to\infty}n^{-1/4}\displaystyle\|X_n(\omega)\|_{\mathrm{op}}=0$;
\item[ii)] $\displaystyle\lim_{n\to\infty}\frac{1}{4n}\Tr (X_n(\omega)^{T}X_n(\omega))=
\sigma^2\in\R$.
\end{itemize}
Then, the sequence of density of rescaled energy levels
\be
\nu_n(E;\omega)=\frac{1}{2^n}\sum_{k=1}^{2^n}\delta\left(\frac{E_{k,n}(\omega)-K_n(\omega)}{\sqrt{n}}-E\right)\label{eq:mean_convergence}
\ee
weakly converges in average, as $n\to\infty$, to a centred Gaussian probability measure with variance $\sigma^2$.
(This means that
\be
\mathbb{E}\int f(E)\de\nu_n(E;\omega)\to\frac{1}{\sqrt{2\pi\sigma^2}}\int f(E)e^{-\frac{E^2}{2\sigma^2}}\de E,
\ee
as $n\to\infty$, for all $f$ bounded and continuous.)
\end{cor}
\begin{proof} The proof is based on the representation of the shifted energy levels in terms of the set of fictitious independent binary variables $r_k$
\be
E_{k,n}(\omega)-K_n(\omega)=\sum_{k=1}^nr_k\lambda_{k,n}(\omega),
\ee
where $\lambda_{k,n}(\omega)$ are the singular values of $X_n(\omega)$.

Let us introduce the sets
\barr
S_{1}&=&\left\{\omega\colon\displaystyle\lim_{n\to\infty}n^{-1/4}\displaystyle\|X_n(\omega)\|_{\mathrm{op}}=0\right\},\\
S_{2}&=&\left\{\omega\colon\text{$\displaystyle\lim_{n\to\infty}\frac{1}{4n}\Tr (X_n(\omega)^{T}X_n(\omega))=\sigma^2$}\right\},\\
S&=&S_1\cap S_2.
\earr
By hypothesis $\mathbb{P}(S_i)=1$ for $i=1,2$, and therefore $\mathbb{P}(S)=1$. Hence, if $\omega\in S$, by Theorem \ref{thm:DEL}
\be
\int e^{itE}\de\nu_n(E;\omega)\to e^{-\frac{\sigma^2t^2}{2}}.
\ee
The above convergence holds $\mathbb{P}$-almost surely (for all $\omega\in S$). Moreover the function $x\mapsto \exp(ix)$ is absolutely bounded and therefore the almost sure convergence can be promoted to convergence in mean
\be
\mathbb{E}\int e^{itE}\de\nu_n(E;\omega)\to e^{-\frac{\sigma^2t^2}{2}}.
\ee
The proof is completed by using L\'evy's continuity theorem.
\end{proof}

Classes of random matrix ensembles which include quantum spin glasses (random two-spin interaction) on generic graphs have been recently considered in \cite{Keating14,Wells14,Keating15,Erdos14}. For these Hamiltonians, using the algebraic identities for Pauli matrices, it has been proved that the limiting spectral density, as the graph cardinality increases, is Gaussian. 
For spin $1/2$'s with nearest neighbourhood random interaction, by the Jordan-Wigner transformation, those systems are equivalent to random quadratic forms of Fermi operators and our method provides an alternative proof of these results. For generic Hamiltonians $\HH_n$, the inverse Jordan-Wigner transformation maps the problem to spin $1/2$ systems with more complicated interactions not considered in previous works. Our method of proof relies on the subset sum structure in the spectrum of quadratic Fermi operators and it is sufficiently robust to be extended to a large class of random Hamiltonians. An exceptional example of random Hamiltonian that does not exhibit a Gaussian limit is presented below.

\begin{example}\label{ex:2} Let us consider the Hamiltonian \eqref{eq:ham_ex1} of Example \ref{ex:1}, but suppose now that $\xi=\xi(\omega)$ is a bounded centred random variables with $0<\mathrm{Var}(\xi^2)<\infty$. Of course $\mu=0$, but $\frac{1}{4n}\Tr (X_n(\omega)^{T}X_n(\omega))=\xi^2(\omega)$ is a random variable.
The rescaled density of states $\nu_n(E;\omega)$ converges $\mathbb{P}$-almost surely to a centred Gaussian density with (random) variance $\xi^2(\omega)$; nevertheless we have no convergence in mean.
\end{example}

In the following sections we discuss explicit examples in detail. In Section \ref{sec:spin} we illustrate the method on spin $1/2$'s systems with fixed nonrandom couplings. We consider in detail the XY model with free boundary conditions and the Ising model with transverse field studied in~\cite{Atas14}.  Then, we present our results for the quantum percolation models and the Anderson models (Section \ref{sec:perc}). In Section \ref{sec:Ginibre} we establish the connection between non-sparse Gaussian quadratic operators and the Ginibre ensemble of random matrix theory. Finally, in Section \ref{sec:band} we apply our theorem to other random band models. The level spacing distribution is discussed in Section~\ref{sec:spacings}.
\section{Spin chains}\label{sec:spin}
As described in the introduction, chains of interacting spin $1/2$'s can be mapped to systems of spinless fermions. We shall apply our theorems to those systems.

The paradigmatic example is provided by the $XY$ chain, a canonical
toy model for quantum spin systems routinely used as a first example to illustrate new concepts. Assuming free boundary conditions, the Hamiltonian of the $XY$-model for $n$ spins can be written as \eqref{eq:HXYfermions}. In this model $A$ and $B$ have a tridiagonal form 
\be
P_nAP_n=\left(
\begin{matrix}
0&1&& & & &0\\
1&0&1&& &\\
&\cdot&\cdot&\cdot&\\
&&\cdot&\cdot&\cdot\\
&&&\cdot&\cdot&\cdot\\
&&&&1&0&1\\
0&&&&&1&0
\end{matrix}
\right),\,
P_nBP_n=\left(
\begin{matrix}
0&\gamma&& & & &0\\
-\gamma&0&\gamma&& &\\
&\cdot&\cdot&\cdot&\\
&&\cdot&\cdot&\cdot\\
&&&\cdot&\cdot&\cdot\\
&&&&-\gamma&0&\gamma\\
0&&&&&-\gamma&0
\end{matrix}
\right);\nonumber
\ee
Note that $A_{ii}=0$ (hence $K_n=0$).
The elementary excitations are  \cite{Lieb61}
\be
\lambda_{k,n}=2\sqrt{1-(1-\gamma^2)\sin^2\theta_{k,n}}, \label{eq:lambda_Lieb}
\ee
where the $\theta_{k,n}$'s are solution of a transcendental equation \cite[Eq. (2.64e)]{Lieb61}.
It is clear that $|\lambda_{k,n}|\leq2$ and
\be
\sigma^2=\lim_{n\to\infty}\frac{1}{4n} \sum_{i,j=1}^n(A_{ij}^2+B_{ij}^2)=\frac{1}{2}(1+\gamma^2),
\ee
and therefore the density of energy levels of the XY model converges to
\be
\frac{1}{2^n}\sum_{k=1}^{2^n}\delta\left(E-E_{k,n}/\sqrt{n}\right)\rightharpoonup\frac{1}{\sqrt{\pi(1+\gamma^2)}}e^{-\frac{E^2}{(1+\gamma^2)}}\de E.\label{eq:DEL_XY}
\ee


Similar considerations can be extended in presence of external fields.  For simplicity we consider the Ising model in transverse field $-\sum_{j=1}^n\sigma_{j}^x\sigma_{j+1}^x+h\sigma_{j}^z$,
where $h\geq0$ is the external magnetic field.  The problem can be reduced to a quadratic form in Fermi operators whose normal modes decomposition has
\be
\lambda_{k,n}=2\sqrt{1-2h\cos\theta_{k,n}+h^2},
\ee
with phases $\theta_{k,n}=\frac{2\pi (k-1)}{n}-\pi,$ ($k=1,\dots,n$) equidistributed. Now $|\lambda_{k,n}|\leq(2+h)$ and 
\be
\lim_{n\to\infty}\frac{1}{4n}\sum_{k=1}^n\lambda_{k,n}^2=(1+h^2).
\ee 
We conclude that
\be
\frac{1}{2^n}\sum_{k=1}^{2^n}\delta\left(E-E_{k,n}/\sqrt{n}\right)\rightharpoonup\frac{1}{\sqrt{2\pi(1+h^2)}}e^{-\frac{E^2}{2(1+h^2)}}\de E,\label{eq:DEL_Ising}
\ee
according to~\cite[Eq. (20)]{Atas14}. Of course, at zero magnetic field $h=0$ we recover the limit density~\eqref{eq:DEL_XY} of the XY model in the Ising limit $\gamma\to1$.

\section{Quantum bond percolation and Anderson models}\label{sec:perc}
A quantum bond percolation model \eqref{eq:percolation} can be cast in the form
\be
\HH_n^{perc}=\sum_{i,j=1}^nA_{ij}(\omega)c_i^{\d}c_j,\label{eq:H_perc_And}
\ee
where $i,j\in V$ denotes the vertices (sites) of a graph $\Gamma=(V,E)$ and $A_{ij}(\omega)=t_{ij}(\omega)1_{(i,j)\in E}$ is the adjacency matrix of $\Gamma$ weighted by random independent Bernoulli variables $\mathbb{P}(t_{ij}=1)=1-\mathbb{P}(t_{ij}=0)=p\in(0,1)$ on a probability space $(\Omega,\mathcal{F},\mathbb{P})$. We assume that the graph $\Gamma$ is a connected regular lattice; in particular, $\Gamma$ does not contain loops (therefore $A_{ii}=0$) and the degree of the vertices is constant $d(i)=d$, where $d(i)$ is the number of neighbours of $i\in V$. ($d$ is called coordination number of the lattice.)

It is well known that the largest eigenvalue of the adjacency matrix of a graph is bounded by the maximal degree. This implies that $\|\sqrt{X_n^TX_n}\|_{\mathrm{op}}\leq d$.  We then compute
\barr
\lim_{n\to\infty}\frac{1}{4n}\Tr (X_n^TX_n)
=\lim_{n\to\infty}\frac{1}{4n}\sum_{\substack{i,j=1\\(i,j)\in E}}^nt_{ij}^2(\omega)
=\lim_{n\to\infty}\frac{1}{4n}\sum_{i=1}^n\sum_{j\in d(i)}^nt_{ij}^2(\omega)=\frac{d p}{4},
\earr
for $\mathbb{P}$-almost all $\omega$. Therefore, by Corollary \ref{thm:DELrandom} we have
\be
\frac{1}{2^n}\sum_{k=1}^{2^n}\delta\left(E-E_{k,n}/\sqrt{n}\right)\rightharpoonup\sqrt{\frac{2}{\pi dp}}e^{-\frac{2E^2}{dp}}\de E.\label{eq:DEL_percolation}
\ee
A very similar analysis can be performed for the Anderson model $\HH_n^{And}$ on a regular lattice defined in Eq. \eqref{eq:Anderson}. The coefficients are $A_{ij}(\omega)=\delta_{ij}v_i(\omega)+t1_{(i,j)\in E}$ and $B_{ij}=0$. The $v_i$'s are i.i.d. variables with mean zero and variance $W^2$. By the strong law of large numbers we find
\barr
\lim_{n\to\infty}\frac{1}{4n}\sum_{i,j=1}^nA_{ij}^2(\omega)=\lim_{n\to\infty}\frac{1}{4n}\left(\sum_{i}^nv_{i}^2(\omega)+\sum_{\substack{i,j=1\\(i,j)\in E}}^nt^2\right)
=\frac{1}{4}\left(W^2+dt^2\right),
\earr
for $\mathbb{P}$-almost all $\omega$.

\section{Gaussian quadratic forms and the Ginibre ensemble} \label{sec:Ginibre}
Let us consider the Hamiltonian~\eqref{eq:quadratic} with random coefficients $A_{ij}(\omega),B_{ij}(\omega)$, $\omega\in\Omega$. We consider the case of $A_{ij}(\omega)=a_{ij}(\omega)/\sqrt{n}$, $B_{ij}(\omega)=b_{ij}(\omega)/\sqrt{n}$ independent Gaussian variables, modulo the symmetries $a_{ij}=a_{ji}$ and $b_{ij}=-b_{ji}$ with mean and variance
\be
\mathbb{E}[a_{ij}]=\mathbb{E}[b_{ij}]=0,\quad
\mathbb{E}[a_{ij}^2]=(1+\delta_{ij})s^2,\quad
\mathbb{E}[b_{ij}^2]=(1-\delta_{ij})s^2.\label{eq:Gauss3}
\ee
For Gaussian random variables the problem is simplified thanks to the following observation: 
if $Z_1,Z_2$ are independent and identically distributed normal variables, then $(Z_1+Z_2)$ and $(Z_1-Z_2)$ are independent normal variables. 
Therefore the entries of the $n\times n$ matrix  $X_n(\omega)=P_n(A(\omega)+B(\omega))P_n$ are $\iid$ Gaussian variables; hence
\be
X_{n}\stackrel{\mathrm{d}}{=}\sqrt{\frac{2s^2}{n}}\mathcal{G},
\ee
where $\mathcal{G}_{ij}$ are $\iid$ standard real Gaussian variable  ($\mathcal{G}$ is a random matrix belonging to the real Ginibre ensemble~\cite{Ginibre65}). We have therefore established that \emph{the elementary excitations $\lambda_{k,n}$ ($k=1,\dots,n$) of a quadratic form with $\iid$ Gaussian coefficients are distributed as the singular values of a real Ginibre matrix $\mathcal{G}(\omega)$ of size $n$ (equivalently, $\lambda_{k,n}^2$ are the eigenvalues of a real $n\times n$ Wishart matrix $\mathcal{W}(\omega)=\mathcal{G}^{T}(\omega)\mathcal{G}(\omega)$)}.

It is well-known that the singular values of $n\times n$ Ginibre matrice whose entries are $\mathcal{O}(1)$ are typically of order $\mathcal{O}(\sqrt{n})$. We therefore rescaled the coefficients $A_{ij},B_{ij}$ by $\sqrt{n}$ to get a sensible limit for the density of energy levels. In fact, using classical asymptotic results on the extreme singular values of random matrices with $\iid$ entries~\cite{Geman80,Bai88}, we know that with probability $1$ all the elementary excitations $\lambda_{k,n}(\omega)$ lie in a fixed interval for large $n$. More precisely we have
\be
\lim_{n\to\infty}\max_{k=1,\dots,n}\lambda_{k,n}(\omega)=\sqrt{2s^2},
\ee
for $\mathbb{P}$-almost all $\omega$.
By the strong law of large numbers we also have
\be
\lim_{n\to\infty}\frac{1}{4n}\Tr(X_n^{T}(\omega)X_n(\omega))=\lim_{n\to\infty}\frac{s^2}{2n^2}\sum_{i,j=1}^n\mathcal{G}_{ij}^2(\omega)=\frac{s^2}{2},
\ee
for $\mathbb{P}$-almost all $\omega$, and by Corollary~\ref{thm:DELrandom} we conclude that for $n\to\infty$
\be
\mathbb{E}\frac{1}{2^n}\sum_{k=1}^{2^n}\delta\left(E-\frac{E_{k,n}-K_{n}}{\sqrt{n}}\right)\rightharpoonup\frac{1}{\sqrt{\pi s^2}}e^{-\frac{E^2}{s^2}}\de E.\label{eq:DEL_Gauss}
\ee
In the rest of this section we use the relation with the Ginibre ensemble to obtain results on the gound state energy and energy gap. The steps of proof are elementary and they borrow the difficult technical statements from previously known results in random matrix theory.
Under the above assumptions of $A_{ij}$ and $B_{ij}$, we have that
\be
\sqrt{\frac{n}{2s^2}}(\lambda_{1,n},\lambda_{2,n},\dots,\lambda_{n,n})\stackrel{\mathrm{d}}{=}(x_1,x_2,\dots,x_n),\label{eq:repres}
\ee
where the joint probability density of the $n$ random variables $x_k$'s is
\begin{align}
&2^nC_n\prod_{i<j}|x_i^2-x_j^2|\prod_ke^{-x_k^2/2}\de x_k,\quad C_n^{-1}=\sqrt{\frac{2^{n^2}}{\pi^n}}\prod_{i=1}^n\Gamma\left(\frac{n-i+1}{2}\right)^2\label{eq:jointlambda}
\end{align}
The joint law~\eqref{eq:jointlambda} is the eigenvalue distribution of the orthogonal chiral ensemble of random matrices. It is usually denoted as chOE, see~\cite[Chapter 3.1]{Forrester10} and~\cite{Verbaarschot00}.
As $n\to \infty$, the empirical distribution of the rescaled variables $x_k/\sqrt{n}$ converges almost surely to the quarter law\cite{MP67,Yin86}
\be
\frac{1}{n}\sum_{k=1}^n\delta\left(x-\frac{x_k}{\sqrt{n}}\right)\to\frac{1}{\pi}\sqrt{4-x^2}1_{(0,2)}(x)\de x\quad\as. \label{eq:quarter}
\ee

From these results we derive now a few properties of the ground state of $\HH_n$. The ground state energy is the lowest level $E_{1,n}$ and we denote by $\Delta_n=E_{2,n}-E_{1,n}$ the ground state energy gap.

\begin{prop}[Ground state energy and ground state energy gap]\label{thm:groundstate} Let $A_{ij}(\omega)$ and $B_{ij}(\omega)$ independent standard Gaussian variables as above (see eq. \eqref{eq:Gauss3}). Then, as $n\to\infty$,
\begin{itemize}
\item[i)] the rescaled ground state energy converges
\be
\frac{3\pi}{(2n)^{3/2}}E_{1,n}\to-s\quad\as;\label{eq:ground_as}
\ee
\item[ii)] The rescaled energy gap $\sqrt{n/2s^2}\Delta_n$ converges in distribution to a random variable whose probability density function is
\be
f(x)=(1+x)e^{-\frac{x^2}{2}-x},\quad x\geq0.\label{eq:limgap}
\ee
\end{itemize}
\end{prop}
By the same proof one shows the almost sure convergence of the rescaled largest energy level $n^{-3/2}E_{2^n,n}$. Therefore, the numerical range of $\HH_n$ is roughly $(-an^{3/2},an^{3/2})$ with $a=(2^{3/2}/3\pi)s$. Note that $\Delta_n=\Ord(n^{-1/2})$. Hence the system is gapless.
\begin{proof}[Proof of Proposition~\ref{thm:groundstate}]
The ground state energy is given by (see \eqref{eq:sp(H)})
\be
E_{1,n}=K_n-\frac{1}{2}\sum_{k=1}^n\lambda_{k,n}.
\ee
By the law of large numbers $n^{-3/2}K_n=n^{-2}\sum_{i=1}^na_{ii}$ converges to zero almost surely. Using \eqref{eq:repres} and the quarter law \eqref{eq:quarter} the following almost sure convergence holds
\be
-\frac{1}{2n}\sum_{k=1}^n\frac{\lambda_{k,n}}{\sqrt{n}}\to-s\int_0^2\frac{\de x}{\pi} x\sqrt{4-x^2}=-\frac{8}{3\pi}s.
\ee
This proves~\eqref{eq:ground_as}.
The ground state energy gap is given by the smallest elementary excitation
\be
\Delta_n=\min_{k=1,\dots,n}\lambda_{k,n}=\sqrt{2s^2/n}\min_{k=1,\dots,n}x_k,
\ee
where $x_1,\dots,x_n$ are distributed according to~\eqref{eq:jointlambda}. The large $n$ distribution of $(n^{-1}\min_{k}x_k^2)$ is given in \cite[Corollary 3.1]{Edelman98}. The claim~\eqref{eq:limgap} follows.
\end{proof}
\begin{rmk}
We expect that one could generalize this analysis to non-Gaussian variables whose first four moments match the Gaussian moments using \emph{Lindenberg exchange strategy.} Using this technique one can replace the Gaussian variables $a_{ij}$ and $b_{ij}$ one at a time by random variables from a desired distribution. This approach is widely used to prove versions of the four moment theorem \cite{tv-review}.
\end{rmk}
\section{Other random band quadratic forms} \label{sec:band}
In this section we show that Corollary \ref{thm:DELrandom} applies to the case when $A$ and $B$ are random band arrays. We introduce a parameter $W_n\geq1$ which  corresponds to the number of non-zero diagonals, i.e. $A_{ij}=B_{ij}=0$ if $|i-j|>W_n$. 

We normalize $A_{ij}(\omega)$, $B_{ij}(\omega)$ to ensure that condition (ii) of Corollary \ref{thm:DELrandom} is satisfied.
To compute the normalization of the matrix entries in terms of $W_n$ we want 
\begin{equation}
\label{eq:normband}
\sigma = \lim_{n\to\infty}\frac{1}{4n}\Tr(X_n^{T}(\omega)X_n(\omega)) =
\lim_{n\to\infty}\frac{1}{4n}\sum_{i, j=1}^n (A_{ij}(\omega)+ B_{ij}(\omega))^2
\end{equation}
to be finite and non-random.
 If there are $W_n$ non-zero diagonals, the matrix $X_n$ has on the order of $nW_n$ non-zero entries, in the sense that we can take $A_{ij}(\omega)=a_{ij}(\omega)/\sqrt{W_n}$, $B_{ij}(\omega)=b_{ij}(\omega)/\sqrt{W_n}$ with $a_{ij}$ and $b_{ij}$ $\iid$ standardized random variables to achieve the finite limit in \eqref{eq:normband}. Here we do not need $a_{ij}$ and $b_{ij}$ to be Gaussian. 
 
We now show that condition (i) of Corollary \ref{thm:DELrandom} is also satisfied and therefore the density of energy levels of random band quadratic forms converges to a Gaussian. Suppose that $W_n=o(n^{1/2})$ and that $a_{ij}$ (and $b_{ij}$) has exponential decay, in the sense that there exists $\delta >0$ such that $\mathbb{E} e^{\delta |a_{ij}|} < \infty$. Then, letting $X_n = P_n(A + B)P_n$, 
\be
\lim_{n\to\infty}n^{-1/4}\displaystyle \|X_n\|_{\mathrm{op}}=0\qquad \mathbb{P}\mathrm{-a.s.}\ .\label{eq:lem}
\ee
We proceed to a proof of \eqref{eq:lem} by showing that for all $L>0$
\be
  \sum_{n=1}^{\infty} \mathbb{P}(n^{-1/4}\displaystyle \|X_n\|_{\mathrm{op}}>L)<\infty, \label{eqBC}
\ee
that implies~\eqref{eq:lem} by the Borel-Cantelli lemma. Note that by triangle inequality $\|X_n\|_{\text{op}} \leq \|P_nAP_n\|_{\text{op}}+\|P_nBP_n\|_{\text{op}} $.
The argument will be identical for the two terms on the right hand side so we will focus on the first one. For a symmetric matrix, the operator norm is equal to the largest modulus of the eigenvalues and it is therefore dominated by any matrix norm. In particular: 
\[
\|P_nAP_n\|_{\text{op}}=\sup_{\|\psi\|_2=1}\langle\psi,A\psi\rangle \leq\sup_{\|\psi\|_1=1}\langle\psi,A\psi\rangle = \max_{1\leq i \leq n} \sum_{j=1}^n |A_{ij}|
\]
Let $Z_i = \sum_{j=1}^n |A_{ij}| = \sum_j |a_{ij}|/\sqrt{W_n}$. Then using that the $Z_i$'s are identically distributed, by the union bound we obtain
 \begin{equation}
 \mathbb{P}(\max_{1\leq i \leq n} Z_i > L) \leq n \mathbb{P}(Z_1 > L) = n \mathbb{P}(\sum_{j=1}^n |a_{1j}| > L \sqrt{W_n}).\end{equation}
  Since $|a_{1j}|$, $j=1,\dots,n$ are $\iid$ random variables, we can apply a Chernoff bound
to get
   \be
   \mathbb{P}(n^{-1/4}\displaystyle \|X_n\|_{\mathrm{op}}>L)\leq 2ne^{C'W_n} e^{-\delta L \sqrt {n^{1/2}W_n}},
   \ee
by which we conclude that~\eqref{eqBC} holds true. For more general sharp concentration inequalities on the operator norm of random matrices see, for instance,~\cite{Bandeira16}.

\section{Level clustering}\label{sec:spacings}
\begin{figure}[t]
\centering
\includegraphics[width=.49\columnwidth]{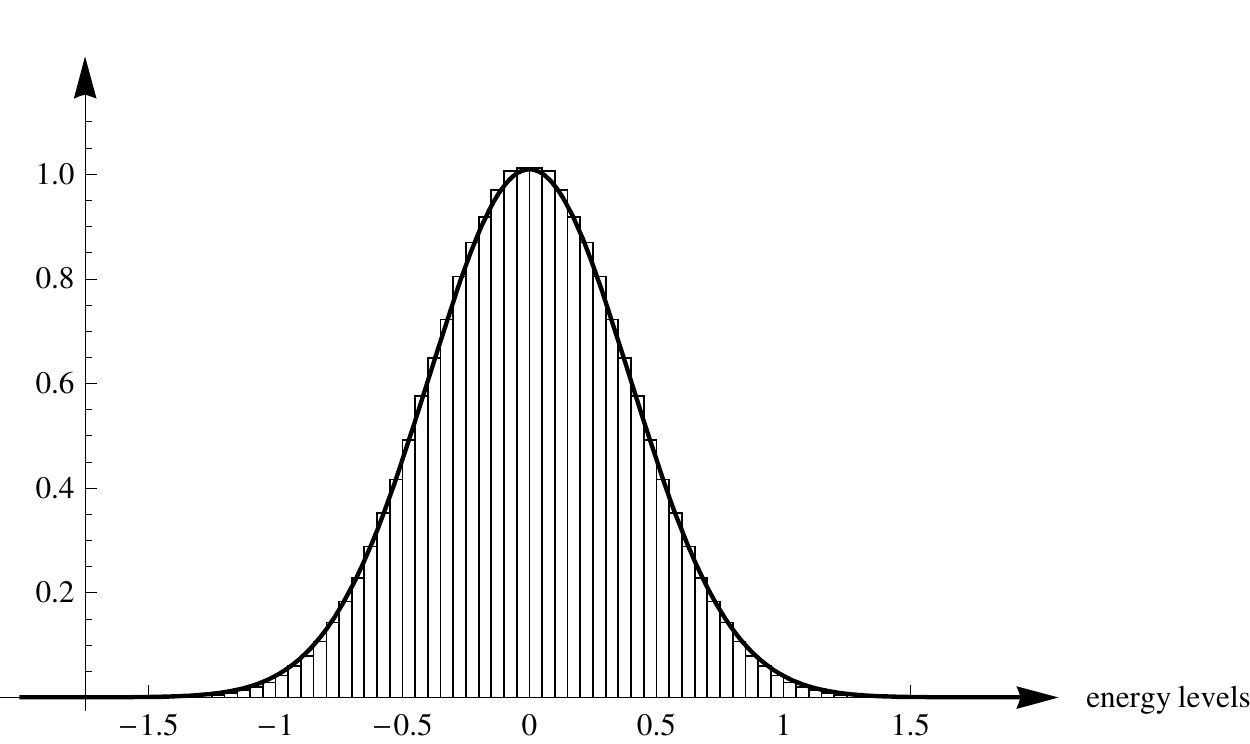}
\includegraphics[width=.49\columnwidth]{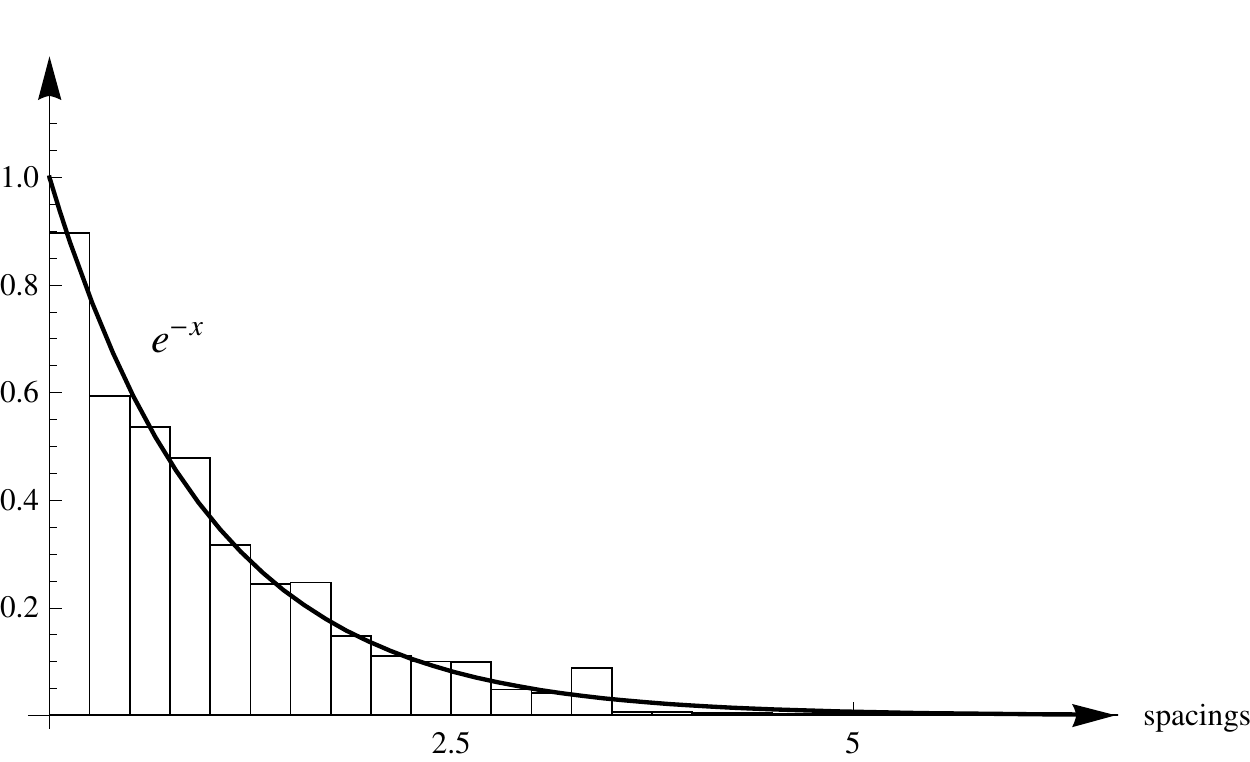}
\caption{\label{fig:XY} XY chain of $n=22$ spins with free ends. Left: Distribution of the rescaled energy levels; the solid line is the limiting Gaussian density~\eqref{eq:DEL_XY}. Right: spacing distribution for the unfolded spectrum; the solid line is the negative exponential $\exp(-x)$ (no fit).}
\end{figure}
\begin{figure}[t]
\centering
\includegraphics[width=.49\columnwidth]{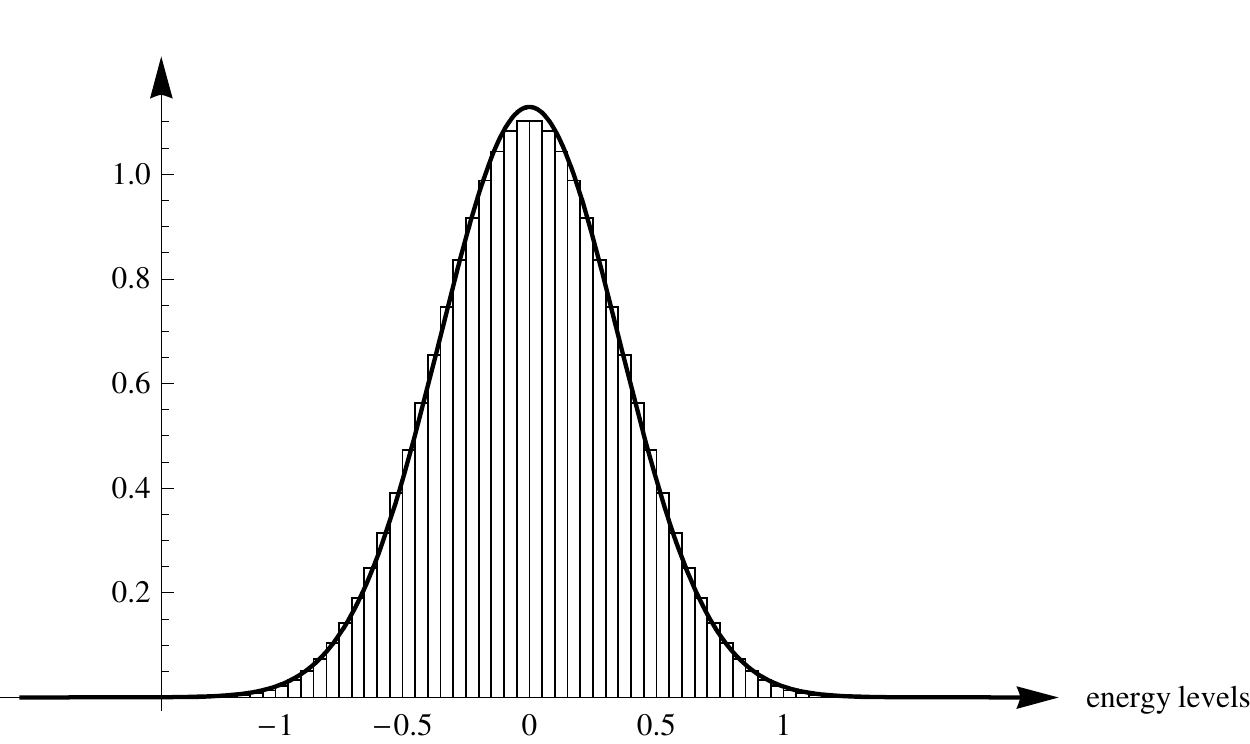}
\includegraphics[width=.49\columnwidth]{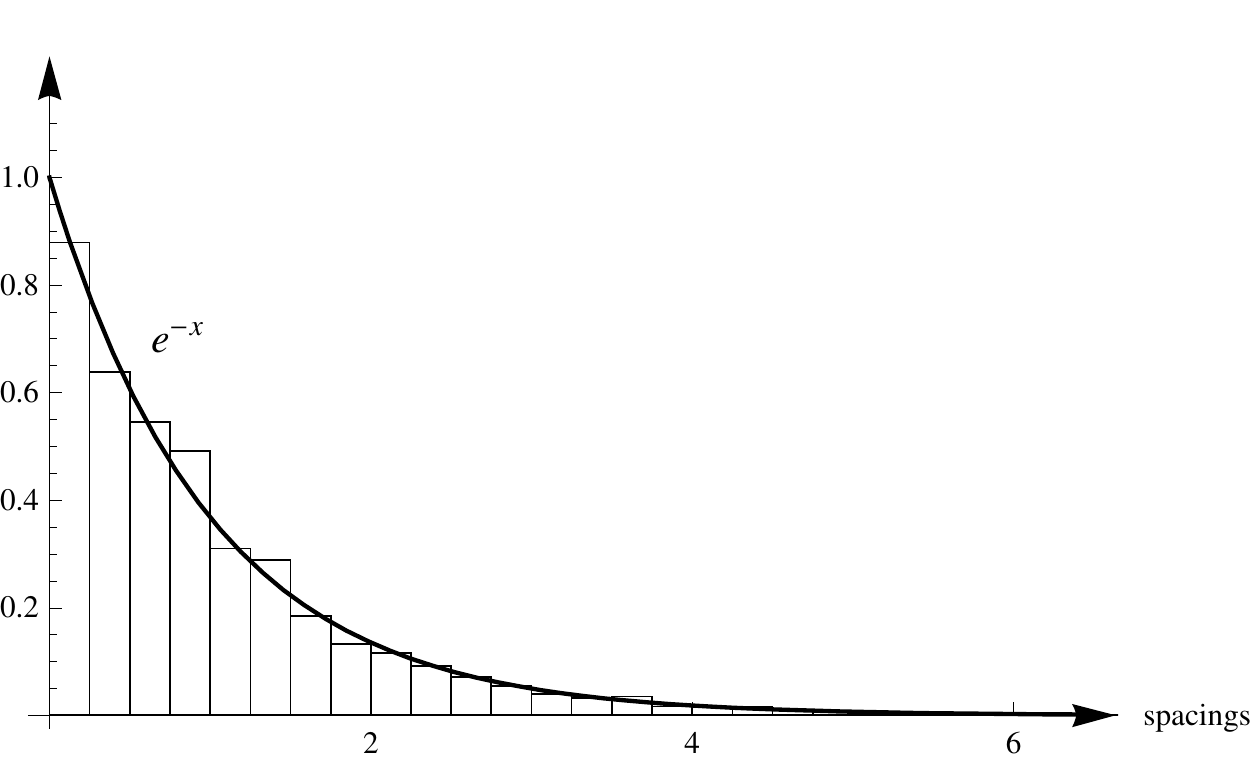}
\caption{\label{fig:Gaussian} Random quadratic form with $\iid$ Gaussian coefficients. Here $n=22$ and $s=1$. Left: Distribution of the  rescaled energy levels;  the solid line is the limiting Gaussian density~\eqref{eq:DEL_Gauss}. Right: spacing distribution for the unfolded spectrum; the solid line is the negative exponential $\exp(-x)$ (no fit).}
\end{figure}
One of the most commonly studied statistical measure of a given spectrum is the level spacing distribution $P(x)$, i.e., the distribution of gaps between consecutive levels. The first step to unravel meaningful information from the spacings is to unfold the spectrum in such a way that the average level spacing in the neighbourhood of each transformed level is unity. In other words, the unfolding procedure is the scaling transformation that removes the irrelevant effects of the varying local mean density. A natural way to unfold the spectrum is by mapping each level $E_{k,n}$ into a new variable $e_{k,n}$ defined as the fraction of energy levels in the spectrum below $E_{k,n}$. In practice, the variation of the density of levels needed for the unfolding is included by fitting the integrated level density or, when explicitly known, by using the limiting level density as an approximation.

We have numerically studied the level spacing distribution for a few instances of quadratic Fermi operators. Fig.~\ref{fig:XY} reports our findings for the XY chain with $n=22$ spins and free boundary conditions. As  illustrated in the left panel, the histogram representing the numerical empirical measure of the energy levels is almost indistinguishable from the limiting Gaussian density. For this reason we have used the limiting Gaussian density in~\eqref{eq:DEL_XY} to unfold the spectrum. $P(x)$ of the unfolded spectrum is shown on the right panel of Fig.~\ref{fig:XY} (we considered about $10^5$ levels in the bulk of the spectrum). We observe that $P(x)$ is maximum at $x=0$ indicating level clustering and it is likely to be the negative exponential $P(x)\simeq e^{-x}$ characteristic of the Poisson process.  We have also studied other spin models obtaining similar results.   This was to be expected since the XY model and its variants are integrable. Poisson statistics have also been numerically observed in previous works for other spin systems integrable by Bethe ansatz, including the Heisenberg chain, the t-J model and the Hubbard model. See, e.g., ~\cite{Poilblanc93}.

We have performed the same investigation for random quadratic forms with independent Gaussian coefficients (see Section~\ref{sec:Ginibre}), where the elementary excitations $\lambda_{k,n}$ of the normal modes are distributed as the singular values of the real Ginibre ensemble~\eqref{eq:repres}. Our findings are reported in Fig.~\ref{fig:Gaussian}. Again, the Gaussian limit~\eqref{eq:DEL_Gauss} is a convincing approximation of the numerical level density even for moderate values of $n$ (left panel). The level spacing in the unfolded spectrum (about $10^5$ levels in the bulk) is well described by a negative exponential. Note that the elementary excitations $\lambda_{k,n}$ repel as the eigenvalues of random matrices (see eq. \eqref{eq:jointlambda}); nevertheless, the energy levels $E_{k,n}$ are given by the subset sums of the $\lambda_{k,n}$'s and this structure dominates the repulsion and enhances the presence of small gaps. At first, this result may be surprising for those working in the field of random matrices or spectral theory of disordered systems. For generic chaotic systems one usually expects level repulsion. We felt natural to provide a theoretical argument to explain the `lack of repulsion' for disordered quasifree fermions.

As argued theoretically by Berry and Tabor~\cite{Berry77}, the energy spectrum of a classically integrable Hamiltonian system represents a sequence of completely uncorrelated numbers and the spectral fluctuations obey Poissonian statistics. The original argument in~\cite{Berry77} is based on the fact that for integrable systems it is possible to perform a canonical transformation into action-angle coordinates. The semiclassical approximation consists in quantizing the action variables so that the quantum energy levels of a classically integrable system are given by the classical Hamiltonian evaluated at points of a lattice (in some cases this quantization rule is exact). Therefore, the level spacings or, more generally, the number statistics of energy levels are related to the problem of counting the number of lattice points enclosed by the Hamiltonian level sets. A computation based on Poisson summation formula then suggests that $P(x)\simeq \exp(-x)$ for generic integrable systems.  This scheme applies only to `generic' systems, and some notable exceptions are quite well known.

Later, this way of reasoning has been extended beyond Hamiltonian mechanics. For instance, the standard argument for Poisson statistics in the case of spin integrable models is as follows~\cite{Poilblanc93}. If a Bethe ansatz holds, the energy levels of the systems are characterised by a set of quasimomenta (that reduce to real momenta for noninteracting spin systems). Typically, these quasi-momenta are the solutions of a set of non-linear equations and therefore the possible quasi-momenta are likely to repel one another, namely they lie on a quasilattice. The level statistics again reduces to the statistics of the lattice positions and the same argument as~\cite{Berry77} leads to Poisson
statistics.

Coming back to the models considered in this paper, we observe that quadratic forms in Fermi operators describe systems of noninteracting Fermi oscillators and are integrable via an exact normal modes decomposition. The existence of the normal modes for quasifree fermions corresponds to the existence of action-angle variables in Hamiltonian mechanics and quasi-momenta in the Bethe ansatz solutions for spin systems. The presence of disorder, e.g., randomness in the parameters, is immaterial regarding the integrability of the model. This explains why spectra of generic quadratic Fermi operators, even with randomness, should follow Poisson statistics.

A more quantitative argument explaining the Poisson statistics for quasifree fermions is based on the idea of `superposition of independent spectra' of Rosenzweig and Porter~\cite{Rosenzweig60} and  Berry and Robnik~\cite{Berry84}.  Note that the Hamiltonian~\eqref{eq:normalmodes} commutes with the number operator $N=\sum_{k}\eta_k^\d\eta_k$ and therefore $\HH_n$ can be block-diagonalized in such a way that each block corresponds to a sector of the Hilbert space with a fixed number $m$ of particles (or number of excited modes), where $m=0,\dots,n$. The sector labeled by $m$ contains $\binom{n}{m}$ eigenstates whose eigenvalues are given by the subset sums over sets of cardinality $m$. In formulae, the level density~\eqref{eq:counting1} can be written as a superposition of $(n+1)$ spectra
\be
\frac{1}{2^n}\sum_{k=1}^{2^n}\delta(E-E_{k,n})=\frac{1}{n+1}\sum_{m=0}^n\mu_n^{(m)}
\ee
where the $m$-particles energy density $\mu_n^{(m)}$ is the normalised counting measure on the $\binom{n}{m}$ energy levels of the $m$-sector. 
The idea now is to compute the gap probability, i.e., the probability of finding no level in a given interval.
Let us consider a large number $L$ of individual spectra $\mu_n^{(m)}$ where $m$, the number of particles, goes off to infinity as $n$ does. 
If one makes the assumption that the individual spectral $\mu_n^{(m)}$ are almost uncorrelated, so that the global gap distribution almost factorizes, using the limit theorem in~\cite{Rosenzweig60,Berry84} one concludes that the gap probability (and hence the level spacing distribution) is given by a negative exponential. We have not been able to carry
out a rigorous analysis of this naive reasoning.

As in the case of Hamiltonian systems, it is not difficult to exhibit exceptional quasifree fermion models deviating from the expected Poisson statistics. One exceptional model is presented below.
\begin{example}
\label{ex:3}
Consider again the model \eqref{eq:ham_ex1}of Examples~\ref{ex:1} and~\ref{ex:2} with $\xi$ fixed or random. One immediately sees that the energy gaps between consecutive levels is constant $E_{k+1,n}-E_{k,n}=\xi$. Hence, the level spacing distribution after the unfolding of the spectrum (neglecting degeneracy of levels) is a delta measure centred at $1$. It is easy to verify that this model does not satisfy the conditions for the limiting theorem on superposition of independent spectra~\cite{Rosenzweig60,Berry84}.

\end{example}

\section*{Acknowledgements}
FDC, AM and FM acknowledge  support  from EPSRC Grant No.\ EP/L010305/1. FDC acknowledges partial support from the Italian National Group of Mathematical Physics (GNFM-INdAM). AM acknowledges the support of the Leverhulme Trust Early Career
Fellowship (ECF 2013-613). FDC is grateful to Michael Bromberg for helpful conversations connected to this project. The authors would like to thank Jens Marklof for useful discussions and Jon P. Keating for his comments on the manuscript.

\end{document}